%% file: main.tex
\documentclass[12pt,a4paper]{article}
\usepackage{amsmath,amsthm,amssymb,amscd,a4wide}
\usepackage{dsfont}
\usepackage{mathrsfs}
\usepackage{setspace}
\newcommand{\ket}[1]{\left| #1 \right>} % for Dirac bras
\newcommand{\bra}[1]{\left< #1 \right|} % for Dirac kets
\input{def.tex}

\numberwithin{equation}{section}

\begin{document}

\thispagestyle{empty}

\begin{center}

{\LARGE\bf Many-particle quantum graphs and\\[5mm]
Bose-Einstein condensation} \\

\vspace*{3cm}

{\large Jens Bolte}%
\footnote{E-mail address: {\tt jens.bolte@rhul.ac.uk}}
{\large and Joachim Kerner}%
\footnote{E-mail address: {\tt j.kerner@rhul.ac.uk}}
\vspace*{2cm}

Department of Mathematics\\
Royal Holloway, University of London\\
Egham, TW20 0EX\\
United Kingdom\\

\end{center}

\vfill

\begin{abstract}
In this paper we propose quantum graphs as one-dimensional models with a
complex topology to study Bose-Einstein condensation and phase transitions in 
a rigorous way. We fist investigate non-interacting many-particle systems on
quantum graphs and provide a complete classification of systems that exhibit 
Bose-Einstein condensation. We then consider models of interacting particles 
that can be regarded as a generalisation of the well-known Tonks-Girardeau gas.
Here our principal result is that no phase transitions occur in bosonic systems
with repulsive hardcore interactions, indicating an absence of Bose-Einstein 
condensation. 
\end{abstract}

\newpage

\input{intro}
\input{1sec}
\input{2sec}
\input{3sec}
\vspace*{0.5cm}
\subsection*{Acknowledgements}
We acknowledge support by the EPSRC network {\it Analysis on Graphs}
(EP/1038217/1). J~K would like to thank the {\it Evangelisches Studienwerk 
Villigst} for financial support through a Promotionsstipendium.

\vspace*{0.5cm}
{\small
\bibliographystyle{amsalpha}
\bibliography{Literature}}

\end{document}

%% file: def.tex
\newcommand{\ud}{\mathrm{d}}

\newcommand{\ue}{\mathrm{e}}

\newcommand{\M}{\mathrm{M}}

\newcommand{\vl}{\boldsymbol{l}}

\newcommand{\vy}{\boldsymbol{y}}

\newcommand{\cD}{{\mathcal D}}
\newcommand{\cE}{{\mathcal E}}

\newcommand{\cH}{{\mathcal H}}

\newcommand{\cV}{{\mathcal V}}

\newcommand{\rz}{{\mathbb R}}
\newcommand{\nz}{{\mathbb N}}
\newcommand{\kz}{{\mathbb C}}

\DeclareMathOperator{\tr}{Tr}

\DeclareMathOperator{\sgn}{sgn}

\newcommand{\eins}{\mathds{1}}

\numberwithin{equation}{section}

\newtheorem{theorem}{Theorem}[section]

\newtheorem{prop}[theorem]{Proposition}

\theoremstyle{definition}
\newtheorem{defn}[theorem]{Definition}
\newtheorem{rem}[theorem]{Remark}

%% file: intro.tex
\section{Introduction}
Bose-Einstein condensation (BEC) is a well established phenomenon in bosonic
many-particle systems. In its original version, a system of non-interacting 
particles in a three-dimensional box was studied and, below a critical 
temperature, condensation into a joint state (the one-particle ground state) 
was found. This is a simple, exactly solvable model for a phase transition. 
The effect, however, depends on the spatial dimension and is absent in 
lower-dimensional, immediate analogues. Examples of BEC in one dimension were 
subsequently found when boundary conditions lead to negative eigenvalues that 
survive the thermodynamical limit \cite{Wilde}. 

BEC in non-ideal gases is a much harder problem, as the Penrose-Onsager 
criterion \cite{PO56} requires a sufficient knowledge not only of the spectrum 
of the interacting many-particle Hamiltonian, but also of the eigenstates. As a 
consequence, results proving condensation in interacting systems and beyond 
mean-field approximations remain scarce \cite{LiebSeiringerProof,LVZ03}.

This paper is the third in a series of papers \cite{BKSingular,BKContact} 
devoted to the investigation of many-particle systems on general compact 
quantum graphs. Quantum graphs are models describing a quantum particle moving 
along the edges of a metric graph. They combine the simplicity of a 
one-dimensional model with the complexity of a graph and have become popular 
models in quantum chaos \cite{KS97,GNUSMY06}. Their spectra display 
correlations that can be well described with random matrix models. 

Models of many particles on a compact graph with singular interactions have 
been developed in the preceding papers \cite{BKSingular,BKContact}. Among them 
is an extension of the well-known Lieb-Liniger model \cite{LL63}, which 
incorporates two-particle $\delta$-interactions, to graphs. In this paper our 
aim is to explore to what extent BEC can or cannot occur in a one-dimensional 
system with complex structure. In a first stage we consider free Bose gases on 
graphs. Depending on the boundary conditions in the vertices the one-particle 
spectrum may or may not contain negative eigenvalues. Relevant for the 
occurrence of BEC are negative eigenvalues that remain negative in the
thermodynamical limit, which we realise in terms of increasing edge lengths.
We provide a complete classification of free Bose gases that display BEC.
In a second step we consider Tonks-Girardeau gases \cite{G60} on graphs as 
limits of Lieb-Liniger models. Our Tonks-Girardeau models describe bosons with 
repulsive hardcore interactions on the edges of a graph. We develop a 
Fermi-Bose map and prove that a Tonks-Girardeau gas is isospectral to a gas of 
free Fermions. This then finally proves the absence of BEC when hardcore 
interactions are switch on, even when BEC was present before.

The paper is organised as follows: In Section~\ref{1sec} we review basic
facts about BEC, quantum graphs, and many-particle systems on graphs.
Section~\ref{2sec} is devoted to the classification of free Bose gases
according to whether or not they display BEC. Tonks-Girardeau gases are
then studied in Section~\ref{3sec} via suitable Fermi-Bose maps, and the
absence of a phase transition is proven.

%% file: 1sec.tex
\section{Preliminaries}
\label{1sec}
In this section we briefly summarise relevant concepts of BEC as well as of 
many-particle quantum graphs. For more details on BEC see 
\cite{PO56,LSSI06,CCGOR11}, on quantum graphs see 
\cite{Kuc04,KS06,GNUSMY06,BolEnd09} and on many-particle quantum graphs see 
\cite{BKSingular,BKContact}.
\subsection{Bose-Einstein condensation}
For a gas of non-interacting bosons (a free gas), it is the macroscopic 
occupation of a one-particle eigenstate that implies the existence of BEC. This 
classical definition is meaningful since in non-interacting systems the 
eigenstates of the full Hamiltonian are (symmetrised) products of one-particle 
eigenstates. However, in the presence of interactions between particles, 
there is no preferred set of one-particle states and it is not immediately 
clear how to generalise this criterion for BEC. A suitable generalisation, 
introduced by Penrose and Onsager \cite{PO56}, is based on the reduced 
one-particle density matrix $\rho_{1}$. At inverse temperature 
$\beta=\frac{1}{T}$ and fixed particle number $N$, the canonical thermal 
density matrix $\rho_{N}$ of the full system is given by
\begin{equation}
\label{fulldensitymatrix}
 \rho_{N}=\frac{1}{Z_N(\beta)}\sum_{n} e^{-\beta E_{n}}\ket{\Psi_{n}}\bra{\Psi_{n}}\ ,
\end{equation}
where $\Psi_{n}$ is the $n$-th eigenvector of the $N$-particle system with 
eigenvalue $E_{n}$, and $Z_N(\beta)=\sum_{n}e^{-\beta E_{n}}$ is the canonical 
partition function. The reduced one-particle density matrix is obtained 
from \eqref{fulldensitymatrix} by tracing out $N-1$ particles, i.e.,
\begin{equation}\label{trace}
 \rho_{1}=N\tr_{2...N}{(\rho_{N})}\ .
\end{equation}
\begin{defn}[Penrose and Onsager]
\label{CritPO} 
Let $\rho_{N}$ be the canonical thermal state \eqref{fulldensitymatrix} of an 
$N$-particle system with one-particle reduced density matrix $\rho_{1}$ as 
in \eqref{trace}. The system is said to display BEC if the largest eigenvalue 
$\lambda_{\max}$ of $\rho_{1}$ satisfies 
\begin{equation}
\label{InequalityBEC}
 c_{1} < \frac{\lambda_{\max}}{N} < c_{2}, \ \ \forall N \geq N_{0}\ ,
\end{equation}
where $0<c_{1}\leq c_{2}$ and $N_{0}$ is sufficiently large. 
\end{defn}
Note that the limit $N \rightarrow \infty$ in Definition \ref{CritPO} is 
accompanied by the limit $V \rightarrow \infty$, where $V$ is the volume of 
the one-particle configuration space, such that the particle density remains 
fixed. This is the standard thermodynamical limit of the canonical ensemble 
\cite{GallavottiBook}. Unfortunately, although the criterion of Penrose and 
Onsager is very general, it is usually difficult to establish BEC rigorously 
in the sense of Definition \ref{CritPO}, since the eigenstates of the full 
system are hard to construct \cite{LSSI06}. On the other hand, it is of 
fundamental interest to understand how particle-particle interactions affect 
the occurrence of BEC \cite{PS86,AP87,BP86}. In this context a more tractable 
approach to BEC aims at identifying phase transitions, which are expected to
occur in any condensation process. This method has been used in 
\cite{BP86,P91,T90}, and it is this approach that will be used below to show 
the absence of BEC in the presence of repulsive hardcore contact interactions
on graphs.

When studying condensation it is more natural to drop the requirement of a
fixed particle number and to work in the grand-canonical ensemble. Then the 
free-energy density at finite volume is given by
\begin{equation}
\label{FreeEnergyDensity}
 f_{V}(\beta,\mu)=-\frac{1}{\beta V}\log Z(\beta,\mu)\ ,
\end{equation}
where
\begin{equation}
 Z(\beta,\mu)=\sum_{N=0}^{\infty}\ue^{N\mu\beta}Z_{N}(\beta)
\end{equation}
is the grand-canonical partition function, and $\mu$ is the chemical potential.
In the grand-canonical ensemble the thermodynamical limit is performed in terms
of the limit $V \rightarrow \infty$ \cite{GallavottiBook} alone. The chemical 
potential is then chosen such that the (fixed) particle density $\rho$ satisfies
\begin{equation}
\label{ConditionDensity}
 \rho = -\frac{\partial f}{\partial \mu}(\beta,\mu)\ ,
\end{equation}
where $f(\beta,\mu)=\lim_{V\to\infty}f_{V}(\beta,\mu)$. Inverting the relation 
\eqref{ConditionDensity} yields a function $\mu=\mu(\rho)$ that then allows 
to replace the chemical potential by the particle density. It may, however, 
happen that the relation \eqref{ConditionDensity} is not invertible for all 
values $\rho \in \rz^{+}$. This is the case in some well-known examples, 
including the three-dimensional free Bose gas and the one-dimensional Bose gas 
with a gap in the one-particle spectrum caused by boundary conditions. In such
a case it is necessary to choose a volume dependent sequence of chemical 
potentials such that the relation \eqref{ConditionDensity}, with 
$f(\beta,\mu)$ replaced by $f_{V}(\beta,\mu_{V})$, is fulfilled at any finite 
volume $V$ \cite{Wilde}. 

In general, a phase transition manifests itself in terms of points where a 
suitable thermodynamical function $g(\beta,\mu)=\lim_{V\to\infty}g_{V}(\beta,\mu)$
is not differentiable \cite{GallavottiMiracleSole,GallavottiBook,Wilde}. In 
the examples mentioned above it is well-known that phase transitions occur 
in this sense and, independently, that BEC occurs \cite{Wilde}. In fact, the 
BEC induces the phase transitions. As a consequence we adopt the point 
of view that an absence of a phase transition indicates an absence of BEC.

Below we shall consider systems of bosons on a graph, interacting via repulsive
hardcore interactions. The models are built after the example of $N$ particles
in one dimension whose interactions are described by the Lieb-Liniger 
Hamiltonian \cite{LL63}
\begin{equation}
\label{LiebLinigerHamiltonian}
 H^{\alpha}_{N}=-\sum_{j=1}^{N} \frac{\partial^{2}}{\partial x_{j}^{2}}+
 \alpha\sum_{i > j} \delta(x_{i}-x_{j})\ .
\end{equation}
In this example repulsive hardcore interactions are obtained by taking the 
limit $\alpha\to\infty$. This procedure leads to Dirichlet boundary conditions 
imposed whenever two coordinates coincide, $x_{i}=x_{j}$. Note that a Bose gas 
in one dimension with repulsive hardcore interactions is known as the 
Tonks-Girardeau gas \cite{YG05,CCGOR11}. The importance of this model lies in 
the fact that, although being a gas of bosons, it exhibits fermionic behaviour 
in various ways \cite{YG05,KWW04}. The origin of the Tonks-Girardeau gas is 
found in the calculation of the classical partition function of a 
one-dimensional gas of hard spheres with diameter $a > 0$ by Tonks 
\cite{Mattis}. Later Girardeau gave a quantum mechanical description of the gas
considered by Tonks and found a one-to-one mapping between a one-dimensional 
gas of bosons with hardcore interaction and a free gas of fermions \cite{G60}. 
Considering the case where $a=0$, Girardeau showed that the spectrum of a gas 
of bosons with hardcore repulsion is the same as the spectrum of a gas of free 
fermions, and the eigenfunctions are related by simple algebraic manipulations.
Let, e.g., $\psi^{(F)}_{0}(x_{1},...,x_{N})$ be the ground state of the gas of 
free fermions whose one-particle configuration space is an interval. The ground 
state $\psi^{(B)}_{0}(x_{1},...,x_{N})$ of a gas of bosons with hardcore point 
interactions is then obtained via the relation
\begin{equation}
\label{bosfermeigen} 
 \psi^{(B)}_{0}(x_{1},...,x_{N})=|\psi^{(F)}_{0}(x_{1},...,x_{N})|\ .
\end{equation}
Despite the close connection between a gas of free fermions and a gas of bosons
with hardcore interactions there are still some subtle differences in the 
eigenfunctions \eqref{bosfermeigen}, see \cite{YG05} for an overview of this 
so-called \textit{Fermi-Bose mapping}. 

For the original Tonks-Girardeau gas BEC was previously investigated 
in various papers \cite{S63,L64,L66,VT79,YG05}. As shown in \cite{PO56}, the 
existence of BEC in the sense of Definition \ref{CritPO} is equivalent to the 
existence of a long-range order in the position representation 
$\rho_{1}(x,x')$ of the reduced one-particle density matrix. For the 
Tonks-Girardeau gas at zero temperature it was shown that
\begin{equation}
 \rho_{1}(x,0) \sim \frac{K}{\sqrt{x}}, \quad x \to \infty\ ,
\end{equation}
where $K>0$ is a constant \cite{YG05,VT79}. This in turn implies that the 
maximal eigenvalue of \eqref{trace} is of order $\sqrt{N}$ and, hence, at zero 
temperature no BEC occurs\cite{YG05}. For finite temperature, however, in 
\cite{L66} an exponential decay of $\rho_{1}(x,x')$ was conjectured. This was
subsequently established in \cite{EL75} at low temperature.

It is interesting to note that studies of particles in one dimension 
interacting via repulsive hardcore interactions have a long history in 
statistical mechanics \cite{vanHove,GallavottiMiracleSole,GallavottiBook}. 
Van Hove, e.g., had shown (see Theorem~5.6.7 of \cite{Ruelle} and 
\cite{vanHove}) that a classical one-dimensional system with hardcore 
two-particle interactions, plus possibly a short-range contribution, shows no 
phase transitions. 
\subsection{Quantum graphs}
The classical configuration space for a particle on a graph is a compact
metric graph, i.e., a finite, connected graph $\Gamma = (\cV,\cE)$ with vertex 
set $\cV = \{v_1,\dots,v_V\}$ and edge set $\mathcal{E}=\{e_1,\dots,e_E\}$. 
The edges are identified with intervals $[0,l_e]$, $e=1,\dots,E$, thus 
assigning lengths to intervals; this then introduces a metric on the graph. 
Note that we do not exclude multiple edges or loops at this point.
\subsubsection{One-particle quantum graphs}
Functions on the graph are collections of functions on the edges,
i.e.,
\begin{equation}
 F=(f_1,\dots,f_E) \ ,\quad\text{with}\quad f_e : [0,l_e]\to\kz\ ,
\end{equation}
so that spaces of functions on $\Gamma$ are (finite) direct sums of
the respective spaces of functions on the edges. The most relevant
space is the Hilbert space
\begin{equation}
\label{HilbertSpaceI}
 \cH_1 = L^2 (\Gamma) := \bigoplus_{e=1}^E L^2 (0,l_e) \ ,
\end{equation}
and all other function spaces are constructed in a similar way.

Standard single particle quantum mechanics suggests to choose the one-particle
Hilbert space $\cH_1 =L^2 (\Gamma)$, see \eqref{HilbertSpaceI}. One-particle 
observables are then self-adjoint operators in $\cH_1$. In the absence
of external forces or gauge fields, the Hamiltonian should be a suitable 
realisation of a Laplacian. As a differential operator the Laplacian acts 
according to
\begin{equation}
 -\Delta_1 F= (-f_1'',\dots,-f_E'')
\end{equation}
on $F\in C^2 (\Gamma)$. We here use the index to indicate that this is a 
one-particle Laplacian.

Viewed as an operator in $L^2 (\Gamma)$ with domain $C_0^\infty (\Gamma)$, this 
Laplacian is symmetric, but not (essentially) self-adjoint. Its self-adjoint
extensions can be classified, and there are several ways to parametrise
these extensions (see, e.g., \cite{KS06,Kuc04}). All of these parametrisations
involve boundary values 
\begin{equation}
\label{fctbv}
 F_{bv} = (f_1(0),\dots,f_E(0),f_1(l_1),\dots,f_E(l_E))
\end{equation}
of the functions in the domain of the operator, as well of derivatives,
\begin{equation}
\label{derbv}
 F'_{bv} = (f'_1(0),\dots,f'_E(0),-f'_1(l_1),\dots,-f'_E(l_E))\ .
\end{equation}
One (unique) characterisation of self-adjoint extensions uses quadratic forms
\cite{Kuc04},
\begin{equation}
\label{quadform1}
 Q_1 [F] = \sum_{e=1}^E\int_0^{l_e}|f'(x)|^2\ \ud x -
 \langle F_{bv},L_1 F_{bv}\rangle_{\kz^{2E}}\ ,
\end{equation}
with domains
\begin{equation}
\label{1formdom}
 \cD_{Q_1} = \{F\in H^1(\Gamma); P_1 F_{bv}=0 \} \ .
\end{equation}
Here $P_1$ is a projector on the space of boundary values $\kz^{2E}$ and $L_1$ 
is a self-adjoint map on $\ker P_1$. (The index indicates that these quantities 
relate to one-particle quantities.) This form is uniquely associated with
a one-particle Laplacian $-\Delta_1$ on the domain
\begin{equation}
 \cD_1(P_1,L_1) = 
 \{F\in H^2(\Gamma); P_1 F_{bv}=0,\ Q_1 F'_{bv}+L_1 Q_1 F_{bv}=0 \} \ ,
\end{equation}
where $Q_1=\eins_{2E}-P_1$.

Any such one-particle Laplacian has a discrete spectrum, with eigenvalues
only accumulating at infinity and following a Weyl asymptotic law (see, e.g.,
\cite{BolEnd09}). Potentially, there are finitely many negative eigenvalues.
Their number is bounded by the number of positive eigenvalues of $L_1$ (in each
case including multiplicities) \cite{KS06}. Hence, a negative semi-definite 
map $L_1$ implies a non-negative Laplacian.
\subsubsection{Many-particle quantum graphs}
\label{subsec:manypart}
An $N$-particle quantum system on a graph requires the tensor product of
one-particle Hilbert spaces, $\cH_N = \cH_1 \otimes ... \otimes \cH_1$.
For a quantum graph this means that
\begin{equation}
\label{NHilbert}
 \cH_N = \Bigl(\bigoplus_{e=1}^E L^2 (0,l_e)\Bigr) \otimes ... \otimes
 \Bigl(\bigoplus_{e=1}^E L^2 (0,l_e)\Bigr) \ ,
\end{equation}
such that vectors $\Psi\in\cH_N$ are collections $\Psi = (\psi_{e_1 e_2...e_{N}})$ 
of $E^N$ functions defined on the hyper-rectangles 
$D_{e_1e_2...e_N}=(0,l_{e_1})\times(0,l_{e_2})\times...\times(0,l_{e_N})$. 
Their disjoint union is denoted as
\begin{equation}
\label{D_Gamdef}
 D_\Gamma^N = \dot{\bigcup_{e_1e_2...e_N}}D_{e_1e_2...e_N} \ ,
\end{equation}
so that one may view $\cH_N$ as
\begin{equation}
\label{HilbertN}
 L^2(D^{N}_\Gamma) = \bigoplus_{e_1e_2...e_{N}}L^2(D_{e_1e_2...e_N})  \ .
\end{equation}
The corresponding Sobolev spaces are introduced in the same way and are 
denoted as $H^m(D^{N}_\Gamma)$. 

A free Hamiltonian for $N$ particles is a lift of a one-particle Hamiltonian
$-\Delta_1$ to $\cH_N$, i.e.,
\begin{equation}
\label{HNfree}
 -\Delta_N^{free} = 
 \sum_{j=1}^N \eins\otimes\dots\otimes(-\Delta_1)\otimes\dots\eins\ ,
\end{equation}
where the (one-particle) Laplacian acts on the coordinates of the $j$-th
particle. As a differential expression, this operator is a Laplacian,
\begin{equation}
 (-\Delta_{N}\Psi)_{e_{1}...e_{N}} =
  -\left(\frac{\partial^{2}}{\partial x_{e_{1}}^{2}}+\dots
  +\frac{\partial^{2}}{\partial x_{e_{N}}^{2}}\right)\psi_{e_{1}...e_{N}}\ ,
\end{equation}
and it can be realised on a suitable domain in \eqref{HilbertN}. Hence, 
any free Hamiltonian \eqref{HNfree} is some extension of the symmetric operator
$-\Delta_N$ defined on the domain $C_0^\infty(D^{N}_\Gamma)$. 

However, there are extensions of $(-\Delta_N,C_0^\infty(D^{N}_\Gamma))$
that are not of the form \eqref{HNfree}; these operators necessarily contain
interactions among the particles. These interactions are induced by boundary
conditions imposed on the functions in their domain and hence are of a
singular type, acting when a particle hits a vertex. A certain class of such 
extensions was introduced in \cite{BKSingular}.

Another class of interactions, introduced in \cite{BKContact}, consists of 
two-particle contact interactions, acting whenever two particles are in the 
same position along an edge. They can be constructed as rigorous versions 
of the formal Hamiltonian
\begin{equation}
\label{formalH_N}
 H_N^\alpha = -\Delta_N +\alpha\sum_{i<j}\delta(x_i-x_j)\ .
\end{equation}
This requires to add additional boundaries to the configuration space
\eqref{D_Gamdef} along diagonal hyperplanes characterised by 
$x_i=x_j$.

A bosonic many-particle system has to be realised on the symmetric
tensor product of $N$ one-particle Hilbert spaces, $\cH_{N,B}=\Pi_B\cH_N$,
where
\begin{equation}
\label{BosonicSymmetry}
 (\Pi_{B}\Psi)_{e_{1}\dots e_{N}} := \frac{1}{N!}\sum_{\pi \in S_{N}}
 \psi_{\pi(e_{1})\dots\pi(e_{N})}(x_{\pi(e_{1})},\dots,x_{\pi(e_{N})})
\end{equation}
is the projector onto the totally symmetric states under particle exchange.

In order to realise the contact interactions indicated in \eqref{formalH_N} 
one has to introduce jump conditions across hyperplanes $x_i=x_j$ on the normal 
derivatives of the functions 
in the operator domain \cite{BKContact}. To achieve this one dissects the 
hyper-rectangles $D_{e_1e_2\dots e_N}$ describing configurations with at least two 
particles on the same edge (i.e., with at least a pair $i\neq j$ such that 
$e_i=e_j$) into polyhedral sub-domains by cutting $D_{e_1e_2\dots e_N}$ along the 
hyperplanes $x_i=x_j$. We denote the dissected version of $D_{e_1e_2\dots e_N}$ 
by $D^\ast_{e_1e_2\dots e_N}$, and the union of all these polyhedral domains by 
$D^{N*}_\Gamma$. The corresponding $L^2$-spaces are 
\begin{equation}
\label{HilbertSpaceDissected}
 L^2(D^{N\ast}_\Gamma)= \bigoplus_{e_1e_2\dots e_{N}}L^2(D^{\ast}_{e_1e_2\dots e_N})\ ,
\end{equation}
and similarly for Sobolev spaces. Given a dissected hyper-rectangle
$D^\ast_{e_1e_2\dots e_N}$, the union of all hyperplanes that form the internal 
boundaries will be denoted as $\partial D^{int}_{e_{1}\dots e_{N}}$, whereas the 
external boundaries $\partial D^{ext}_{e_1e_2\dots e_N}$ are the boundaries of the 
undissected hyper-rectangle $D_{e_1e_2\dots e_N}$. We then set
\begin{equation}
\label{intbound}
 \partial D^{N,ext/int}_\Gamma =
 \bigcup_{e_{1}\dots e_{N}}\partial D^{ext/int}_{e_{1}\dots e_{N}} \ .
\end{equation}
Implementing jump conditions across $\partial D^{N,int}_\Gamma$ will then yield 
the $\delta$-interactions \eqref{formalH_N}. 

One can also introduce boundary conditions on $\partial D^{N,ext}_\Gamma$ that
prevent the resulting $N$-particle Laplacian from being a free Hamiltonian
\cite{BKSingular}. These boundary conditions cause the interactions 
mentioned above to arise when one particle hits a vertex. In order to 
implement such interactions one needs the boundary values 
\begin{equation}
\label{BVI}
 \Psi_{bv}(\vy) = \begin{pmatrix} \sqrt{l_{e_{2}}\dots l_{e_{N}}} 
 \psi_{e_{1}\dots e_{N}}(0,l_{e_{2}}y_{1},\dots,l_{e_{N}}y_{N-1}) \\ 
 \sqrt{l_{e_{2}}\dots l_{e_{N}}} 
 \psi_{e_{1}\dots e_{N}}(l_{e_{1}},l_{e_{2}}y_{1},\dots,l_{e_{N}}y_{N-1}) \end{pmatrix} \ ,
\end{equation}
and
\begin{equation}
\label{BVII}
 \Psi^{'}_{bv}(\vy) = \begin{pmatrix} \sqrt{l_{e_{2}}\dots l_{e_{N}}}  
 \psi_{e_{1}\dots e_{N},x^{1}_{e_{1}}}(0,l_{e_{2}}y_{1},\dots,l_{e_{N}}y_{N-1}) \\ 
 -\sqrt{l_{e_{2}}\dots l_{e_{N}}} 
 \psi_{e_{1}\dots e_{N},x^{1}_{e_{1}}}(l_{e_{1}},l_{e_{2}}y_{1},\dots,l_{e_{N}}y_{N-1})  
 \end{pmatrix} \ ,
\end{equation}
where $\vy=(y_{1},\dots,y_{N-1})\in [0,1]^{N-1}$, of functions 
$\Psi\in H^1_B(D^N_\Gamma)$ and the normal derivatives. 
Acting on these (vertex related) boundary values are the bounded and 
measurable maps $P_{N},L_{N}: [0,1]^{N-1} \to \M(2E^{N},\kz)$ such that for 
a.e.\  $\vy\in [0,1]^{N-1}$ the linear map $P_{N}(\vy)$ is an orthogonal 
projector and $L_{N}(\vy)$ is a self-adjoint endomorphism on $\ker P_{N}(\vy)$. 
These maps define two bounded, self-adjoint multiplication operators $\Pi_{N}$ 
and $\Lambda_{N}$, respectively, on $L^{2}(0,1) \otimes \kz^{2E^{N}}$
\cite{BKContact}. Note that in \cite{BKSingular} it was shown that actual 
interactions are obtained whenever $P_{N},L_{N}$ are not all independent of 
$\vy$. 

Following \cite{BKContact}, the quadratic form for $N$ bosons on a graph 
subject to any of the interactions introduced above is   
\begin{equation}
\label{QuadFormContactI}
\begin{split}
 Q^{(N)}_{B}[\Psi] 
  &= N \sum_{e_{1}\dots e_{N}}\int_{0}^{l_{e_{1}}}\dots\int_{0}^{l_{e_{N}}} 
    |\psi_{e_{1}\dots e_{N},x_{e_1}}(x_{e_{1}},\dots,x_{e_{N}})|^{2}\ \ud x_{e_N}\dots
    \ud x_{e_1} \\
  &\quad -N\int_{[0,1]^{N-1}}\langle\Psi_{bv},L_{N}(\vy)\Psi_{bv} 
    \rangle_{\kz^{2E^{N}}} \ud\vy \\
  &\quad +\frac{N(N-1)}{2}\sum_{e_{2}...e_{N}}\int_{[0,1]^{N-1}} \alpha(y_1)\ 
    |\sqrt{l_{e_{2}}\dots l_{e_{N}}}\psi_{e_{2}e_{2}\dots e_{N}}
    (l_{e_2}y_1,\vl\vy)|^{2}\ \ud\vy \ ,
\end{split}
\end{equation}
where $\vl\vy=(l_{e_2}y_1,l_{e_3}y_2,\dots,l_{e_N}y_{N-1})$, with domain
\begin{equation}
\begin{split}
\label{QNformdomain}
 \cD_{Q^{(N)}_{B}} = \{\Psi \in H^{1}_{B}(D^{N\ast}_\Gamma);\ P_{N}(\vy)
 \Psi_{bv}(\vy)=0\ \text{for a.e.}\ \vy\in [0,1]^{N-1}\}\ .
\end{split}
\end{equation}
The second line on the right-hand side of \eqref{QuadFormContactI} implies the 
vertex-related interactions, whereas the third line yields contact
interactions with (bounded) variable strength $\alpha$. The latter can be turned
into a repulsive hardcore interaction by taking the limit $\alpha\to\infty$.
Taking this limit one has to amend the domain \eqref{QNformdomain} in that
the functions have to vanish on the internal boundaries 
$\partial D^{N,int}_\Gamma$. This subspace of $H^{1}_{B}(D^{N\ast}_\Gamma)$ shall be 
denoted as $H^{1}_{0,int,B}(D^{N\ast}_{\Gamma})$.

Under these conditions it was shown in \cite{BKContact} that the quadratic form
\eqref{QuadFormContactI} is closed and semi-bounded. Hence, there exists a
unique, self-adjoint and semi-bounded operator associated with this form.
This operator is a self-adjoint realisation of the $N$-particle Laplacian which 
will, in the case of hardcore interactions, be denoted by 
$(-\Delta_{N},\cD^{\alpha=\infty}_{N}(P_N,L_N))$. It has a discrete spectrum and the
eigenvalue asymptotics follows a Weyl law \cite{BKContact}.

%% file: 2sec.tex
\section{BEC in non-interacting Bose gases}
\label{2sec}
As described in  \eqref{HNfree}, a self-adjoint realisation of the free 
$N$-particle Laplacian follows from a tensor product construction  
based on a given one-particle Laplacian $(-\Delta_{1},\cD_1 (P_1,L_1))$.
The eigenfunctions $\{\Psi_{n}\}_{n \in \nz^N}$ of the $N$-particle Laplacian 
are then given as symmetrised products of the one-particle eigenfunctions 
$\{\phi_n\}_{n\in\nz}$, i.e.,
\begin{equation}
\Psi_{n}=\Pi_{B}(\phi_{n_{1}} \otimes\dots\otimes \phi_{n_{N}})  \ ,
\end{equation}
where $n=(n_1,\dots,n_N)$. The $N$-particle eigenvalues are 
$\lambda_{n}=k^{2}_{n_{1}}+\dots+k^{2}_{n_{N}}$, 
where $\{k^{2}_{n}\}_{n \in \nz}$ are the corresponding one-particle eigenvalues.

We shall mainly work in the grand-canonical ensemble where the thermodynamical 
limit is taken as the limit of an infinite volume of the one-particle 
configuration space (see Section~\ref{1sec}). For a graph the volume of the 
one-particle configuration space is the total length of the graph,  
\begin{equation}
\label{Length}
 \mathcal{L}=\sum_{e=1}^{E}l_{e} \ .
\end{equation}
In principle, the infinite-volume limit can be achieved by either increasing
the number of edges, or by scaling the lengths of the edges. As we do not want 
to change the topology of the graph, we here choose to leave the number of 
edges fixed and only increase the edge lengths by a common factor.
\begin{defn}
\label{ThermoLimit}
Let $\Gamma$ be a compact, metric graph with edge lengths
$l_1,\dots, l_E$. The {\it thermodynamical limit (TL)} consists of the 
scaling
\begin{equation}
\label{ScalingThermoLimit}
 l_{e} \mapsto \eta l_{e}\ , 
\end{equation}
and taking the limit $\eta\to\infty$. 
\end{defn}
We shall also use the notation
\begin{equation}
\lim_{TL} F(\mathcal{L})
\end{equation}
for the thermodynamical limit of a function $F(\mathcal{L})$.

According to general folklore, in the absence of disorder a free gas 
of bosons in one dimension shows no BEC at finite temperature. This, 
however, is only true if the one-particle spectrum has no gap separating a 
finite number of eigenvalues at the bottom of the spectrum from the (quasi-) 
continuum of states above. In order to generate such a gap in a quantum graph, 
it is necessary for the Laplacian $-\Delta_{1}$ to possess negative eigenvalues.

An upper bound for the number $n_-(-\Delta_1(P_1,L_1))$ of negative eigenvalues
of the one-particle Laplacian was proved in \cite{KS06}. The exact number
was later determined in \cite{BL10}, where the matrix
\begin{equation}
 M_{0}(l_1,\dots,l_E) := 
 \begin{pmatrix} m_{1}(l_1) & & 0 \\ &\ddots& \\ 0& &m_{E}(l_E)\end{pmatrix}\ ,
\end{equation}
with
\begin{equation}
 m_{e}(l_e) := \frac{1}{l_{e}}\begin{pmatrix}-1 & 1 \\ 1 & -1 \end{pmatrix}\ ,
\end{equation}
was introduced. It was then shown in \cite{BL10} that 
\begin{equation}
\label{numbernegeig}
 n_-(-\Delta_1(P_1,L_1)) = n_+(L_1+Q_1M_0 Q_1)\ ,
\end{equation}
where the right-hand side denotes the number of positive eigenvalues of
the linear map $L_1+Q_1M_0 Q_1$ on $\ker P_1\subseteq\kz^{2E}$ (and
$Q_1=\eins_{2E}-P_1$). Therefore, when the edge lengths tend to infinity in the 
TL, the quantity $n_-(-\Delta_1(P_1,L_1))$ approaches $n_+(L_1)$. This, however, 
does not imply that there are $n_+(L_1)$ negative Laplace-eigenvalues in the 
TL as some of the negative eigenvalues could approach zero in the limit. 
Nevertheless, for the question of BEC the number of positive eigenvalues of 
the map $L_1$ is still relevant, as we shall show below.

In the following Proposition we first prove absence of BEC when the domain
of the one-particle Laplacian is defined in terms of a negative semi-definite 
map $L_1$, and therefore the Laplacian has no negative eigenvalues.  
\begin{prop}
\label{NonInter}
Let $-\Delta_{N}$ be a bosonic, self-adjoint realisation of the free 
$N$-particle Laplacian. If the corresponding one-particle Laplacian
$(-\Delta_{1},\cD_1 (P_1,L_1))$ is such that $L_1$ is negative semi-definite, 
no Bose-Einstein condensation occurs at finite temperature in the 
thermodynamical limit. 
\end{prop}
\begin{proof} 
It is well known that there is no BEC in a gas of free bosons on an interval
of finite length with either Dirichlet or Neumann boundary conditions. In both
cases the eigenvalues are known explicitly and the standard argument applies. 
The same is true for a compact quantum graph with Dirichlet or Neumann 
boundary conditions: the eigenvalue equations on the edges decouple and, 
again, the eigenvalues are known explicitly. The spectrum, therefore, is the 
union of the spectra for each edge. 

Let
\begin{equation}
 \mathcal{N}(K)=\#\{n;\ k^{2}_{n} \leq K^{2} \}
\end{equation}
be the eigenvalue counting function for $(-\Delta_{1},\cD_1 (P_1,L_1))$,
where the eigenvalues $k_n^2$ are counted with their multiplicities,
and denote by $\mathcal{N}_{D/N}(K)$ the respective counting functions for 
the Dirichlet- and Neumann-Laplacian. A bracketing argument then implies
\begin{equation}
\label{inclusion}
 \mathcal{N}_{D}(K) \leq \mathcal{N}(K) \leq \mathcal{N}_{N}(K)\ .
\end{equation}
This follows, e.g., from the proof of Proposition 4.2 in \cite{BolEnd09}, 
taking into account that $L_1$ is negative semi-definite.

In the grand-canonical ensemble the expected number of particles can be 
expressed as
\begin{equation}
 N(\beta,\mu) = \sum_{n=0}^\infty \frac{1}{e^{\beta(k_n^{2}-\mu)}-1}
 =\int_0^\infty\frac{1}{\ue^{\beta(k^2 -\mu)}-1}\ \ud\mathcal{N}(k)\ .
\end{equation}
The relation \eqref{inclusion} hence implies that
\begin{equation}
\label{Nbracket}
 N_{D}(\beta,\mu) \leq N(\beta,\mu) \leq N_{N}(\beta,\mu) \ .
\end{equation}
In the TL the expected particle density is
\begin{equation}
\label{rhobrack}
 \rho(\beta,\mu)=\lim_{TL}\frac{N(\beta,\mu)}{\mathcal{L}}\ ,
\end{equation}
and \eqref{Nbracket} implies
\begin{equation}
 \rho_{D}(\beta,\mu) \leq \rho(\beta,\mu) \leq \rho_{N}(\beta,\mu) \ .
\end{equation}
Due to the explicit knowledge of the eigenvalues it is known that 
$\rho_{D}(\beta,\mu)=\rho_{N}(\beta,\mu)$. Thus, \eqref{rhobrack}
yields
\begin{equation}
 \rho(\beta,\mu) = \rho_{D/N}(\beta,\mu) \ .
\end{equation}
Therefore, as the Dirichlet- and Neumann case shows no BEC, the same holds for
any gas of free bosons satisfying the conditions of the Proposition.
\end{proof} 

It is known \cite{Wilde} that in dimension less than three, despite the general
folklore a free Bose gas may show BEC, if the spectrum of the one-particle 
Hamiltonian has a gap below zero. More precisely, if the one-particle
spectrum is such that there are infinitely many positive eigenvalues and,
say, one negative eigenvalue that remains negative after having taken the
TL, infinitely many particles will occupy the eigenstate
corresponding to the negative eigenvalue (the ground state). An example for 
this mechanism is given by an attractive delta-potential on the real axis.
This has exactly one negative eigenvalue and undergoes BEC at low 
temperatures \cite{DELTA76}. Furthermore, the condensate is spatially 
localised around the location of the delta-potential. It is worth mentioning 
that in the current context the real axis with a delta-potential at the origin 
is a quantum graph with two edges of infinite length and one vertex, at which 
appropriate boundary conditions are imposed. 

We now determine a class of quantum graphs that maintain a spectral gap below
zero in the thermodynamical limit. Our main tool will be a Rayleigh quotient,
\begin{equation}
 R[\Psi] = \frac{Q_1[\Psi]}{\|\Psi\|^2_{L^{2}(\Gamma)}}\ ,\qquad\Psi\in\cD_{Q_1}\ ,
\end{equation}
which is an upper bound for the ground state eigenvalue.
\begin{prop}
\label{PropGround}
Let $\Gamma$ be a compact, metric graph with a self-adjoint one-particle 
Laplacian $(-\Delta_{1}, \cD_{1}(P_1,L_1))$. Assume that $L_{1}$ has at least one
positive eigenvalue and denote the largest eigenvalue by $L_{\max}$. Then the 
ground state eigenvalue $-\kappa_{\min}^2<0$ of the one-particle Laplacian 
converges to $-L_{\max}^2<0$ in the thermodynamical limit.
\end{prop}
\begin{proof}
As $L_1$ is assumed to possess at least one positive eigenvalue, 
$n_+(L_1)\geq 1$, the relation \eqref{numbernegeig} implies that the
Laplacian has at least one negative eigenvalue as long as the edge lengths 
are finite. Hence, for any $\Phi\in\cD_{Q_1}$,
\begin{equation}
\label{groundstateest}
 -s^2 \leq - \kappa_{\min}^2 \leq R[\Phi]\ .
\end{equation}
Here $-s^2$ is the lower bound for the spectrum of the one-particle Laplacian 
proved in \cite{KS06}, where $s$ a solution of
\begin{equation}
 s\tanh\left(\frac{sl_{\min}}{2}\right) = L_{\max}\ ,
\end{equation}
and is $l_{\min}$ the shortest edge-length. In the TL, where 
$l_{\min}\to\infty$, the lower bound in \eqref{groundstateest} converges to 
$-L_{\max}^2$. To find an upper bound in \eqref{groundstateest} we need to 
determine the Rayleigh quotient of a suitable trial function.

We assume that $P_{1}\neq\eins_{2E}$ as this would correspond to Dirichlet 
boundary conditions in the vertices, where it is known that there are no 
negative eigenvalues. Hence, there exists a non-trivial vector 
\begin{equation}
 v:=(c_1,\dots,c_E,c_{E+1},\dots,c_{2E})^T\in\ker P_1 \ .
\end{equation}
Using the components of such a vector, we now define a trial function $\Phi$ 
with components
\begin{equation}
 \phi_{e}(x) = \begin{cases} 
 c_e\bigl(1-\frac{x}{\lambda}\bigr)^\alpha &,\  x \leq\lambda \\ 
 0 &,\ \lambda\leq x \leq l_e-\lambda \\
 c_{e+E}\bigl(\frac{x}{\lambda}+1-\frac{l_e}{\lambda}\bigr)^\alpha &,\
 x \geq l_e -\lambda \end{cases}\ , \qquad \alpha\geq 1\ .
\end{equation}
As we shall take the TL, given any value for $\lambda$ we can arrange that 
$l_e\geq 2\lambda$ for all $e=1,\dots,E$. The boundary values of this 
function, therefore, are
\begin{equation}
 \Phi_{bv} = (c_1,\dots,c_E,c_{E+1},\dots,c_{2E})^T =v\in\ker P_1\ ,
\end{equation}
hence this function is in the domain \eqref{1formdom} of the quadratic form.

We now intend to estimate the Rayleigh quotient of $\Phi$, noting that we are 
free to choose $v\in\ker P_1$. The optimal choice for our purpose is to let
$v=\Phi_{bv} $ be an eigenvector of $L_1$ corresponding to its maximal 
eigenvalue $L_{\max}>0$. Then,
\begin{equation}
\begin{split}
 Q_1 [\Phi] 
 &= \sum_{e=1}^E\int_0^{l_e}|\phi'_e (x)|^2\ \ud x - 
    \langle\Phi_{bv},L_1\Phi_{bv}\rangle \\
 &= \frac{\alpha^2}{(2\alpha-1)\lambda}\sum_{e=1}^E 
    \bigl(|c_e|^2+|c_{e+E}|^2\bigr) - L_{\max}\|\Phi_{bv}\|^2_{\kz^{2E}} \\
 &= \left(\frac{\alpha^2}{(2\alpha-1)\lambda}-L_{\max}\right)\,
    \|\Phi_{bv}\|^2_{\kz^{2E}}\ .
\end{split}
\end{equation}
Moreover,
\begin{equation}
 \|\Phi\|^{2} = \sum_{e=1}^E\int_0^{l_e}|\phi_e (x)|^2\ \ud x 
 = \frac{\lambda}{2\alpha+1}\sum_{e=1}^E\big(|c_e|^2 +|c_{e+E}|^2\bigr) 
 = \frac{\lambda}{2\alpha+1}\,\|\Phi_{bv}\|^2_{\kz^{2E}} \ ,
\end{equation}
so that
\begin{equation}
 R[\Phi] = \left(\frac{\alpha^2}{(2\alpha-1)\lambda}-L_{\max}\right)\,
           \frac{2\alpha+1}{\lambda} \ .
\end{equation}
The right-hand side is negative when 
$\lambda > \frac{\alpha^2}{(2\alpha-1)L_{\max}}$ and has a minimum at 
$\lambda_{\min}=\frac{2\alpha^2}{(2\alpha-1)L_{\max}}$. With this optimal choice 
we find that
\begin{equation}
 R[\Phi] = -\frac{4\alpha^2-1}{4\alpha^2}\,L_{\max}^2\ .
\end{equation}
As $\alpha\geq 1$ can be chosen arbitrarily large in the TL, the optimal upper 
bound in \eqref{groundstateest} approaches $-L_{\max}^2$. Hence, 
$-\kappa_{\min}^2$ converges to $-L_{\max}^2$ in the TL.
\end{proof}
We are now in a position to state our main result of this section.
\begin{theorem}
\label{BECI}
Let a free Bose gas be given on a quantum graph with a one-particle 
Laplacian $(-\Delta_{1}, \cD_{1}(P_1,L_1))$ such that $L_1$ has at least one
positive eigenvalue. Then, in the thermodynamical limit, there is a critical 
temperature $T_c>0$ such that Bose-Einstein condensation occurs below $T_c$. 
\end{theorem}
\begin{proof}
We denote the non-negative eigenvalues of the one-particle Laplacian (counted
with their multiplicities) as $k_0^2\leq k_1^2\leq k_2^2\leq k_3^2\leq\dots$.
In the grand-canonical ensemble, the expected particle number occupying
states of non-negative energy is
\begin{equation}
 N_+(\beta,\mu) = \sum_{n=0}^\infty \frac{1}{e^{\beta(k_n^{2}-\mu)}-1}\ .
\end{equation}
Recall that by Proposition~\ref{PropGround} the one-particle ground state 
eigenvalue is a distance $L_{\max}^2>0$ below zero. Hence the chemical potential 
$\mu$ has to satisfy $\mu\leq -L_{\max}^2$. The density of particles in states 
of non-negative energy in the TL then is
\begin{equation}
\label{NumberParticles}
 \rho_+(\beta,\mu)=\lim_{TL}\frac{N_+(\beta,\mu)}{\mathcal{L}}\ .
\end{equation}
In order to evaluate this expression we employ Proposition~$5.2$ of 
\cite{BolEnd09}, which provides a  preliminary form of the trace formula and can 
be rearranged as
\begin{equation}
\label{TraceFormula}
\begin{split}
 \sum_{n=0}^{\infty}h(k_{n}) 
  &= \frac{\mathcal{L}}{2\pi}\int_{-\infty}^\infty h(k)\ \ud k + \gamma h(0)
     - \frac{1}{4\pi}\int_{-\infty}^\infty h(k)\,s(k)\ \ud k \\
  &\quad + \sum_{l\neq 0} \frac{1}{4\pi i}\int_{-\infty}^{\infty} 
      \tr[\Lambda(k)U^{l}(k)]\,h(k)\ \ud k\ .
\end{split}
\end{equation}
Here $\gamma$ is a constant related to the multiplicity of the eigenvalue zero,
$\Lambda,U$ are matrix-valued functions involving the boundary conditions,
and $s$ is another function related to the boundary conditions. In this trace
formula, $h$ is a test function from a suitable test function space 
\cite{BolEnd09}. 

Now, choosing $h(k)=\frac{1}{\ue^{\beta(k^{2}-\mu)}-1}$, the left-hand side of 
\eqref{TraceFormula} is $N_+(\beta,\mu)$, and the right-hand side provides
four separate contributions to $N_+(\beta,\mu)$. It is obvious that the second
and the third term give no contributions to $\rho_+(\beta,\mu)$. An estimate
of the fourth term can be found in the proof of Theorem~$5.4$ in 
\cite{BolEnd09},
\begin{equation}
 \sum_{l\neq 0} \left|\int_{-\infty}^{\infty} \tr[\Lambda(k)U^{l}(k)]\,h(k)\ \ud k 
 \right| = O\bigl(\ue^{-\sigma l_{\min}}\bigr)\ ,
\end{equation}
where $l_{\min}$ is the shortest edge-length and $\sigma>0$. Hence, as in the 
thermodynamical limit $l_{\min}\to\infty$, this term, too, gives no contribution 
to $\rho_+(\beta,\mu)$. Therefore, the only non-vanishing contribution
comes from the first term (which also provides the Weyl term in the
asymptotics of the eigenvalue count),
\begin{equation}
 \rho_+(\beta,\mu) = \frac{1}{\pi}\int_0^\infty\frac{1}{\ue^{\beta(k^{2}-\mu)}-1}
 \ \ud k = \frac{1}{\sqrt{4\pi\beta}}\,g_{\frac{1}{2}}(\ue^{\beta\mu})\ .
\end{equation}
Here
\begin{equation}
 g_{\nu}(z) = \frac{1}{\Gamma(\nu)}\int_{0}^{\infty}
  \frac{x^{\nu-1}}{z^{-1}e^{x}-1}\ \ud x = \sum_{k=1}^\infty\frac{z^k}{k^\nu}
\end{equation}
is the well-known Bose-Einstein function (a polylogarithm). The series 
converges for $|z|<1$ and has a finite limit as $z\to 1$ when $\nu>1$. Here, 
however, $z=\ue^{\beta\mu}\leq\ue^{-\beta L_{\max}^2}<1$ as $\mu\leq-L_{\max}^2$. 
Hence, $\rho_+(\beta,\mu)$ is finite for all $\beta>0$ and tends to zero as 
$\beta\to\infty$ (i.e., $T\to 0$).

The total particle density, $\rho(\beta,\mu)$, also has a contribution 
from particles occupying states with negative energy,
\begin{equation}
 \rho(\beta,\mu) = \rho(\beta,\mu)_- + \rho(\beta,\mu)_+ \ .
\end{equation}
Given that the limiting particle density has a fixed value, $\rho_0$, below a 
certain critical temperature $T_c=\frac{1}{\beta_c}$ the negative energy states
must be populated because $\rho(\beta,\mu)_+<\rho_0$ when $\beta>\beta_c$.
This critical temperature is implicitly defined by
\begin{equation}
 \rho_0 =\frac{1}{\sqrt{4\pi\beta_c}}\,g_{\frac{1}{2}}(e^{-\beta_c L_{\max}^2})\ .
\end{equation}
More explicitly, when $T\leq T_c$,
\begin{equation}
 \rho_- = \rho_0 - \rho_+ = \rho_0\left( 
 1 - \frac{1}{\rho_0}\frac{1}{\sqrt{4\pi\beta}}\, 
 g_{\frac{1}{2}}(e^{-\beta L_{\max}^2})\right) \ .
\end{equation}
Hence, below the critical temperature the relative occupation of the negative 
energy states is
\begin{equation}
 \frac{\rho_-}{\rho_0} =  1 - \sqrt{\frac{\beta_c}{\beta}}
 \frac{g_{\frac{1}{2}}(e^{-\beta L_{\max}^2})}{g_{\frac{1}{2}}(e^{-\beta_c L_{\max}^2})} >0\ .
\end{equation}
Thus BEC occurs for $T<T_C$, with the relative occupation of the negative
energy states approaching one as $T\to 0$. 
\end{proof}

%% file: 3sec.tex
\section{BEC in Bose gases with contact interactions}
\label{3sec}
We now consider bosons on a graph with many-particle interactions described 
by the quadratic form \eqref{QuadFormContactI}. As explained in 
Section~\ref{subsec:manypart}, these interactions consists of two types: 
they are either located in the vertices of the graph, or they are contact 
interactions. Introducing the latter type of interactions amounts to
generalising the original Lieb-Linger model \cite{LL63} to graphs. 
In the limit of an infinite strength such a model turns into a gas with
hardcore repulsion and can be viewed as a generalisation of a Tonks-Girardeau 
gas to the graph setting. We shall now investigate under what 
circumstances BEC on graphs may or may not occur in the presence of repulsive 
hardcore interactions.
\subsection{Fermi-Bose mapping on general quantum graphs}
Our first goal is to generalise the Fermi-Bose mapping introduced by 
Girardeau~\cite{G60,YG05} to general quantum graphs. This requires us to 
introduce systems of $N$ free fermions, and to relate them to systems of $N$ 
bosons interacting via repulsive hardcore interactions. The fermionic states 
are vectors in the fermionic $N$-particle Hilbert space 
$L^2_F(D_\Gamma^N)=\Pi_F L^2(D_\Gamma^N)$, where 
\begin{equation}
\label{FermonicSymmetry}
 (\Pi_F\Psi)_{e_{1}\dots e_{N}} := \frac{1}{N!}\sum_{\pi \in S_{N}}(-1)^{\sgn\pi}
 \psi_{\pi(e_{1})\dots\pi(e_{N})}(x_{\pi(e_{1})},\dots,x_{\pi(e_{N})})\ .
\end{equation}
An analogous notation is used for Sobolev spaces. Interactions will be 
described in terms of a quadratic form
\begin{equation}
\label{FormFermions}
\begin{split}
 Q^{(N)}_{F}[\Psi] 
  &= N \sum_{e_{1}\dots e_{N}}\int_{0}^{l_{e_{1}}}\dots\int_{0}^{l_{e_{N}}} 
    |\psi_{e_{1}\dots e_{N},x_{e_1}}(x_{e_{1}},\dots,x_{e_{N}})|^{2}\ \ud x_{e_N}\dots
    \ud x_{e_1} \\
  &\quad -N\int_{[0,1]^{N-1}}\langle\Psi_{bv},L_{F,N}(\vy)\Psi_{bv} 
    \rangle_{\kz^{2E^{N}}} \ud\vy 
\end{split}
\end{equation}
with domain
\begin{equation}
\begin{split}
\label{QNformdom}
 \cD_{Q^{(N)}_{F}} = \{\Psi \in H^{1}_{F}(D^{N}_\Gamma);\ P_{F,N}(\vy)
 \Psi_{bv}(\vy)=0\ \text{for a.e.}\ \vy\in [0,1]^{N-1}\}\ .
\end{split}
\end{equation}
Here we use the same notation as in Section~\ref{subsec:manypart},
however with an additional index $F$ indicating the fermionic nature
of the form.
Using the methods of \cite{BKContact}, one can readily show that the form 
\eqref{FormFermions} is closed and semi-bounded, hence it corresponds to a 
unique self-adjoint operator, $(-\Delta_{N},\cD^{N}_{F}(P_{F,N},L_{F,N}))$. 
Furthermore, a standard bracketing argument implies the discreteness of the 
spectrum of this self-adjoint operator as well as a Weyl law for its
eigenvalue asymptotics. We denote the set of all such fermionic $N$-particle
Laplacians by $\mathcal{M}_{F,N}$.

On the other hand, $N$ Bosons with repulsive hardcore interactions are
described in terms of a quadratic form
\begin{equation}
\label{QuadFormHard}
\begin{split}
 Q^{(N)}_{B}[\Psi] 
  &= N \sum_{e_{1}\dots e_{N}}\int_{0}^{l_{e_{1}}}\dots\int_{0}^{l_{e_{N}}} 
    |\psi_{e_{1}\dots e_{N},x_{e_1}}(x_{e_{1}},\dots,x_{e_{N}})|^{2}\ \ud x_{e_N}\dots
    \ud x_{e_1} \\
  &\quad -N\int_{[0,1]^{N-1}}\langle\Psi_{bv},L_{B,N}(\vy)\Psi_{bv} 
    \rangle_{\kz^{2E^{N}}} \ud\vy 
\end{split}
\end{equation}
on the domain
\begin{equation}
\begin{split}
\label{QNformharddomain}
 \cD_{Q^{(N)}_{B}} = \{\Psi \in H^{1}_{0,int,B}(D^{N\ast}_\Gamma);\ P_{B,N}(\vy)
 \Psi_{bv}(\vy)=0\ \text{for a.e.}\ \vy\in [0,1]^{N-1}\}\ ,
\end{split}
\end{equation}
see the paragraph below \eqref{QNformdomain}. We recall that
$H^{1}_{0,int,B}(D^{N\ast}_{\Gamma})\subset H^{1}_{B}(D^{N\ast}_\Gamma)$ consists of the
functions vanishing along the internal boundary $\partial D^{N,int}_\Gamma$ of 
the dissected configuration space. Again, we added an index $B$ to reflect the 
bosonic nature of the form. Consequently, we denote the associated operator by 
$(-\Delta_{N},\cD^{\alpha=\infty}_{B,N}(P_{B,N},L_{B,N}))$, and the set
of all such operators by $\mathcal{M}^{\alpha=\infty}_{B,N}$.
\begin{theorem}
\label{Fermion-boson-mapping_II}
There exists a bijective map 
\begin{equation}
\label{BFmapdef}
 \sigma: \mathcal{M}_{F,N}\rightarrow\mathcal{M}^{\alpha=\infty}_{B,N}\ ,
\end{equation}
such that the operators $(-\Delta_{N},\cD^{N}_{F}(P_{F,N},L_{F,N}))$ and
$\sigma[(-\Delta_{N},\cD^{N}_{F}(P_{F,N},L_{F,N}))]$ are isospectral. This
map can be constructed explicitly.
\end{theorem}
\begin{proof}
In order to introduce the Fermi-Bose map \eqref{BFmapdef} we first note that
the fermionic symmetry implies the vanishing of functions in 
$H^1_F(D^{N}_\Gamma)$ along the internal boundary $\partial D^{N,int}_\Gamma$, 
see \eqref{intbound}. Therefore, we seek to relate $H^1_F(D^{N}_\Gamma)$ 
to the bosonic Sobolev space $H^1_{0,int,B}(D^{N\ast}_\Gamma)$. More explicitly, 
we define a bijective map 
\begin{equation}
 T_\sigma:H^{1}_{F}(D^{N}_\Gamma) \rightarrow H^{1}_{0,int,B}(D^{N\ast}_\Gamma)
\end{equation}
as follows: Let $\Phi_{F}=(\varphi^{F}_{e_{1}\dots e_{N}})\in H^{1}_{F}(D^{N}_\Gamma)$,
and divide the components into classes such that representatives have the
same set of edge indices $e_{1},\dots,e_{N}$, up to permutations. Given a fixed 
representative $\varphi^{F}_{e_{1}\dots e_{N}}$, let $n(e)$ be the number of times an
edge $e\in\mathcal{E}$ occurs among the edge indices and introduce (particle)
labels $\zeta(1),\dots,\zeta(n(e))$. Then define a subdomain 
$\Omega\subset D_{e_1\dots e_N}$ such that all 
$x =(x^{\zeta(1)}_e ,\dots, x^{\zeta(n(e))}_e)\in\Omega$ fulfil
\begin{equation}
x^{\zeta(1)}_{e} < \dots < x^{\zeta(n(e))}_{e}\ .
\end{equation}
This is used to define the component $\varphi^{B}_{e_1\dots e_N}$ of a bosonic 
state by setting
\begin{equation}
\label{bosdef}
 \varphi^{B}_{e_1\dots e_N}(x):=\varphi^{F}_{e_{1}\dots e_{N}}(x)
\end{equation}
for $x\in\Omega$, and extending this to all of $D_{e_{1}\dots e_{N}}$ using the
bosonic symmetry. Finally, by permuting the edge indices of 
$\varphi^{B}_{e_{1}\dots e_{N}}$ and assigning the same values \eqref{bosdef} to
each representative we obtain all other components, defining a symmetric 
function $\Phi_{B}=T_\sigma(\Phi_F)\in H^{1}_{0,int,B}(D^{N\ast}_\Gamma)$. This 
construction can be reversed in an obvious way so that the map $T_\sigma$ is 
invertible.

Based on the map $T_\sigma$ we introduce a diagonal matrix $\Sigma(\vy)$ 
with non-vanishing entries $\Sigma_{ee}(\vy)\in\{1,-1\}$ that take account 
of the possible sign changes introduced by $\sigma$. Note that this matrix
is such that $\Sigma(\vy)^2=\eins$. More explicitly, if 
$\Phi \in \cD_{Q^{(N)}_{F}}$ is a function with boundary values $\Phi_{bv}$, the 
function $T_\sigma(\Phi)\in H^{1}_{0,int,B}(D^{N\ast}_\Gamma)$ has boundary values
\begin{equation}
 [T_\sigma(\Phi)]_{bv}(\vy)=\Sigma(\vy)\Phi_{bv}(\vy)\ .
\end{equation}
Furthermore, we set
\begin{equation}
\label{MappingMatrizen}
 P^{\sigma}_{N}(\vy)=\Sigma(\vy) P_{N}(\vy)\Sigma(\vy)\ \qquad\text{and}\qquad
 L^{\sigma}_{N}(\vy)=\Sigma(\vy)L_{N}(\vy)\Sigma(\vy)\ ,
\end{equation}
which are, for every $\vy\in[0,1]^{N-1}$, projectors and self-adjoint maps on
$\ker P^{\sigma}_{N}$, respectively. Hence, given a fermionic quadratic form 
$(Q^{(N)}_{F},\cD_{Q^{(N)}_{F}})$ as in \eqref{FormFermions} and \eqref{QNformdom}, 
we can associate to it a unique bosonic form via
\begin{equation}
\label{EqualityForms}
\begin{split}
 Q^{(N)}_{F}[\Phi]
   &= \sum_{e_{1}\dots e_{N}}\int_{0}^{l_{e_{1}}}\dots\int_{0}^{l_{e_{N}}}|\nabla
      \varphi_{e_{1}\dots e_{N}}|^{2} \ \ud x^{1}_{e_{1}}\dots\ud x^{N}_{e_{N}}\\
   &\qquad -N \int_{[0,1]^{N-1}} \langle \Phi_{bv},L_{N}(\vy)\Phi_{bv}
      \rangle_{\kz^{2E^{N}}} \ \ud\vy \\
   &= \sum_{e_{1}\dots e_{N}}\int_{0}^{l_{e_{1}}}\dots\int_{0}^{l_{e_{N}}}|\nabla
      T_\sigma(\varphi)_{e_{1}\dots e_{N}}|^{2} \ \ud x^{1}_{e_{1}}\dots\ud x^{N}_{e_{N}} \\
   &\qquad -N \int_{[0,1]^{N-1}} \langle [T_\sigma(\Phi)]_{bv},L^{\sigma}_{N}(\vy)
      [T_\sigma(\Phi)]_{bv}\rangle_{\kz^{2E^{N}}} \ \ud \vy \\
   &= Q^{(N)}_{B}[T_\sigma(\Phi)]
\end{split}
\end{equation}
on the domain
\begin{equation}
\label{FermBosdom}
 T_\sigma(\cD_{Q^{(N)}_F})=\{\Phi\in H^{1}_{0,int,B}(D^{N\ast}_{\Gamma});\ 
    P^{\sigma}_{N}(\vy)\Phi_{bv}(\vy)=0 \ \text{for a.e.} \ \vy 
    \in [0,1]^{N-1}\} \ .
\end{equation}
The associated self-adjoint operator is of hardcore type, see 
\eqref{QuadFormHard} and \eqref{QNformharddomain}.

The above explicit contruction defines the Fermi-Bose map \eqref{BFmapdef}
through 
\begin{equation}
 \sigma[(-\Delta_{N},\cD^{N}_{F}(P_{F,N},L_{F,N}))]=
  (-\Delta_{N},\cD^{\alpha=\infty}_{B,	N}(P^{\sigma}_{N},L^{\sigma}_{N})) \ .
\end{equation}
By construction it is bijective, and due to \eqref{EqualityForms} the operators 
are isospectral.
\end{proof}
\subsection{Bose-Einstein condensation in a gas of bosons interacting via 
repulsive hardcore interactions}
In a second step we use the Fermi-Bose mapping established in 
Theorem~\ref{Fermion-boson-mapping_II} in order to study BEC in a system of 
particles interacting via repulsive hardcore interactions. 
With this goal in mind we first consider the fermionic realisations of the 
$N$-particle Laplacian described in the previous Subsection and compare their 
free-energy densities \eqref{FreeEnergyDensity} with those of free fermion 
gases. For the latter we choose two comparison operators. 

The first reference model is that of free fermions with Dirichlet boundary
conditions in the vertices. For every $N\in\nz$ we hence choose 
$P_{F,N}^D=\eins_{2E^{N}}$ and $L_{F,N}^D=0$ with corresponding operator
$(-\Delta_{N}, \cD^{N}_{F}(\eins_{2E^{N}},0))$. This is a textbook example
(see, e.g., \cite{SchwablSM}) for which the free-energy density 
\eqref{FreeEnergyDensity} is well known to be
\begin{equation}
\label{SequenceDirichlet}
\begin{split}
 f_{F,D}(\beta,\mu) 
  &= -\lim_{TL}\frac{1}{\mathcal{\beta L}}\sum_{n=0}^{\infty}\log{\left(
     1+e^{-\beta(k^{2}_{n}-\mu)}\right)} \\
  &= -\frac{1}{\pi \beta}\int_{0}^{\infty}\log{\left(1+e^{-\beta(k^{2}-\mu)}\right)}
     \ \ud k\ .
\end{split}
\end{equation}
Here $\{k^{2}_{n}\}_{n \in \nz_{0}}$ are the one-particle eigenvalues. Note that 
this function is smooth, 
$f_{F,D}\in C^{\infty}\left((0,\infty) \times \rz \right)$, hence there is no 
phase transition in a gas of free fermions.

The second reference model also describes free fermions, however with
standard Robin boundary conditions in the vertices. Here $P_{F,N}^R=0$ and 
$L_{F,N}^R=M\eins_{2E^{N}}$, and the corresponding operator is 
$(-\Delta_{N}, \cD^{N}_{F}(0,M\eins_{2E^{N}}))$, where $M>0$ is a suitable
constant.
\begin{prop}
\label{Fermions} 
Let $(-\Delta_{N},\cD^{N}_{F}(P_{F,N},L_{F,N}))_{N\in\nz}$ be a family of fermionic 
Laplacians indexed by the particle number $N$ as introduced above. Assume 
that for this family there exists $M>0$ such that
\begin{equation}
\label{ConditionProof}
 \|\Lambda_{F,N}\|_{op} \leq M\ , \qquad \forall N \ .
\end{equation}
Then the grand-canonical free-energy density $f_{F}(\beta,\mu)$ coincides
with the free-energy density \eqref{SequenceDirichlet} of free fermions with 
Dirichlet boundary conditions in the vertices.
\end{prop}
\begin{proof} 
Using the min-max principle \cite{ReeSim78} in the same way is in 
Proposition~\ref{NonInter}, we conclude that 
\begin{equation}
\label{EquationProof}
 f^{\mathcal{L}}_{F,R}(\beta,\mu) \leq f^{\mathcal{L}}_{F}(\beta,\mu) \leq 
 f^{\mathcal{L}}_{F,D}(\beta,\mu)
\end{equation}
holds for any family of the fermionic Laplacians that we allow. Since 
$f^{\mathcal{L}}_{F,R}(\beta,\mu)$ is the free-energy density of a free fermion 
gas, it can be reduced to 
\begin{equation}
\label{EquationProofFreeEnergy}
 f^{\mathcal{L}}_{F,R}(\beta,\mu) = 
 -\frac{1}{\mathcal{\beta L}}\sum_{K^{2}_{n} \leq 0}\log{\left(
 1+e^{-\beta(K^{2}_{n}-\mu)}\right)} - 
 \frac{1}{\mathcal{\beta L}}\sum_{K^{2}_{n} > 0}\log{\left(
 1+e^{-\beta(K^{2}_{n}-\mu)}\right)}\ ,
\end{equation}
where $\{K^{2}_{n}\}_{n \in \nz_{0}}$ are the one-particle eigenvalues.

The number of negative eigenvalues of the one-particle Laplacian is 
finite so that the first term on the right-hand side of 
\eqref{EquationProofFreeEnergy} does not contribute in the TL. The second 
term can be evaluated using the trace formula in the same way as in the proof 
of Theorem~\ref{BECI}. This then gives
\begin{equation}
 \lim_{TL}f^{\mathcal{L}}_{F,R}(\beta,\mu)=f_{F,D}(\beta,\mu)\ .
\end{equation}
Together with the bracketing \eqref{EquationProof} this completes the proof.
\end{proof}
\begin{rem} Note that condition \eqref{ConditionProof} can be understood as 
a stability condition for the interaction potential similar to what is 
often required in statistical mechanics \cite{LiebBound,Ruelle}. More 
precisely, one requires a self-adjoint $N$-particle Hamiltonian $\hat{H}_{N}$ 
to be bounded from below by $-NB$ where $B \geq 0$ is some constant. 
\end{rem}
Finally, we can state the main results of this section. 
\begin{theorem}
\label{TheoremBEC}
Let $(-\Delta_{N},\cD^{\alpha=\infty}_{N}(P_{B,N},L_{B,N}))_{N\in\nz}$ be a family of 
bosonic Laplacians with repulsive hardcore interactions, indexed by the 
particle number $N$. Assume that for this family there exists $M>0$ such that
\begin{equation}
 \|\Lambda_{B,N}\|_{op} \leq M\ , \qquad \forall N \ .
\end{equation}
Then the associated bosonic grand-canonical free-energy density 
$f_{B}(\beta,\mu)$ coincides with the free-energy density 
\eqref{SequenceDirichlet} of free fermions with Dirichlet boundary conditions 
in the vertices,
\begin{equation}
 f_{B}(\beta,\mu) = -\frac{1}{\pi \beta}\int_{0}^{\infty}\log{\left(
 1+e^{-\beta(k^{2}-\mu)}\right)}\ \ud k\ .
\end{equation}
This function is smooth and, hence, there occurs no phase transition.
\end{theorem}
\begin{proof}
Using the inverse Fermi-Bose map as described in
Theorem~\ref{Fermion-boson-mapping_II}, we associate to the family 
$(-\Delta_{N},\cD^{\alpha=\infty}_{N}(P_{B,N},L_{B,N}))_{N\in\nz}$ of bosonic Laplacians
an isospectral family of fermionic Laplacians. According to 
Proposition~\ref{Fermions} the resulting family of fermionic Laplacians has a 
free-energy density $f_{F}(\beta,\mu)=f_{F,D}(\beta,\mu)$, and due to the 
isospectrality with the bosonic family one immediatley finds that
$f_{B}(\beta,\mu)=f_{F,D}(\beta,\mu)$.
\end{proof}
\begin{rem} Theorem~\ref{TheoremBEC} can be regarded as a quantum 
statistical version of the Theorem by van Hove \cite{vanHove}.
\end{rem}